\documentclass[twocolumn]{svjour3}

\makeatletter
\def\@seccntformat#1{\@ifundefined{#1@cntformat}%
   {\csname the#1\endcsname\quad}  
   {\csname #1@cntformat\endcsname}
}
\let\oldappendix\appendix 
\renewcommand\appendix{%
    \oldappendix
    \newcommand{\section@cntformat}{\appendixname~\thesection\quad}
}
\makeatother

\usepackage{amsmath,amssymb, amsthm}
\usepackage{stmaryrd}
\usepackage[utf8]{inputenc}
\usepackage{algorithmicx,xpatch}
\usepackage{algorithm}
\usepackage{algpseudocode}
\usepackage{pgf-pie}
\usepackage{listings}
\usepackage{tikz}
\usepackage{setspace}
\usepackage{multirow}
\usepackage{url}
\usetikzlibrary{arrows,shapes,positioning,shadows,trees}
\usepackage{caption}
\usepackage{moresize}

\title{Efficient and secure modular operations using the Adapted Modular Number System}
\titlerunning{Efficient and secure modular operations using AMNS}

\begin{document}
\newcommand\comment[1]{{\color{black} #1}}

\author{Laurent-St\'ephane Didier \and
Fangan-Yssouf Dosso  \and
           Pascal V\'eron
}


\institute{L.-S. Didier, F.-Y. Dosso, P. V\'eron \at
              Institut de
  Math\'ematiques de Toulon\\ Universit\'e de Toulon, France \\
              \email{didier@univ-tln.fr}\\
              \email{dosso@univ-tln.fr}\\
              \email{veron@univ-tln.fr}
}

%
%

\maketitle              

\begin{abstract}
The Adapted Modular Number System\break (AMNS) is a sytem of representation of integers to speed up arithmetic operations modulo a prime $p$. Such a system can be defined by a tuple $(p, n, \gamma, \rho, E)$ where $E\in\mathbb{Z}[X]$. In \cite{amns_12} conditions are given to build AMNS with $E(X)=X^n + 1$. In this paper, we generalize their results and show how to generate multiple AMNS for a given prime $p$ with $E(X)=X^n-\lambda$ and $\lambda\in\mathbb{Z}$. Moreover, we propose a complete set of algorithms without conditional branching to perform arithmetic and conversion operations in the AMNS, using a Montgomery-like method described in \cite{amns_08}.

\end{abstract}

\section{Introduction}

Efficient implementations of most of modern public-key cryptography algorithms depend on 
the efficiency of the modular arithmetic implementation. Such cryptosystems  usually need fast arithmetic modulo 
integers of size from 160 bits up to several thousand bits. Adapting or building unusual arithmetic for cryptographic purpose may offer interesting properties.  

For instance, Residue Number Systems \cite{gar_56} is a non positional arithmetic which have parallel properties. This makes them suitable for SIMD architectures \cite{abs_12} and adapted to many cryptosystems \cite{bi_04,behz_16,abs_12}. Furthermore, they have interesting leak resistant properties \cite{bilt_04,bde_13}.

In positional (usual) number system, a positive integer $N$ is represented in base $\beta$ as follows : 
\begin{center}
	$N = \sum\limits_{i=0}^{k-1} d_i\beta^i$,
\end{center}
where $\,0 \leqslant d_i < \beta$. If $d_{k-1} \neq 0$, $N$ is called a $k$-digit $\beta$-base number. $\beta$ is often taken as a power of two.

In modular arithmetic, the modulus $p$ is generally used several times and according the target architecture, the base $\beta$ is fixed. In such context, elements modulo $p$ can be seen as polynomials of degree lower than the number of digits of $p$ in base $\beta$.
In \cite{amns_04}, Bajard and al. introduced the Modular Number System (MNS) which can be seen as a generalization of positional number systems. 
\begin{definition}
	\label{def:mns}
	A modular number system (MNS) $\mathcal{B}$ is defined by a tuple $(p, n, \gamma, \rho)$, such that for every integer $0 \leqslant x < p$, there exists a vector $V = (v_0, \dots, v_{n-1})$ such that: 
	\begin{center}
		{\large $x \equiv \sum\limits_{i=0}^{n-1} v_i\gamma^i \, ({\text{\rm mod}} \, p) \;$},
	\end{center}
	with $|v_i| < \rho$ \, and \, $0 < \rho,\gamma < p$. \\
	In that case, we say that $V$ (or equivalently the polynomial $V(X) = v_0 + v_1.X + \dots + v_{n-1}.X^{n-1}$) is a representative of $x$ in $\mathcal{B}$ and we notate $V \equiv x_{\mathcal{B}}$.
\end{definition}
\begin{example}
	In Table \ref{tab:exple1}, we give a representative in the MNS $\mathcal{B} = (17, 3, 7, 2)$ for each element modulo $p = 17$. 
	In particular, we can verify that if we evaluate $-1+X+X^2$ in $\gamma$ , we have $-1 + \gamma + \gamma^2 = - 1 + 7 + 49 = 55 \equiv 4 \, \text{\rm mod} \, 17$. We have also $deg(-1+X+X^2) = 2 < 3$ and $\|-1+X+X^2\|_{\infty} = 1 < 2$.
\begin{table}[H]
\renewcommand{\arraystretch}{1.4}
\begin{center}
\small
\begin{tabular}{|c|c|c|c|c|}
\cline{1-4} 
0 & 1 & 2 & 3 \\
\cline{1-4}
0 & 1 & $-X^2$ & $1-X^2$ \\
\cline{1-4}
\multicolumn{1}{c}{} \\
\cline{1-4}
4 & 5 & 6 & 7\\
\cline{1-4}
$-1+X+X^2$ & $X+X^2$ & $-1+X$ & $X$ \\
\cline{1-4}
\multicolumn{1}{c}{} \\
\cline{1-4}
8 & 9 & 10 & 11 \\
\cline{1-4}
$1+X$ & $-X-1$ & $-X$ & $-X+1$ \\
\cline{1-4}
\multicolumn{1}{c}{} \\
\hline
12 & 13 & 14 & 15 & 16\\
\hline
$-X-X^2$ & $1-X-X^2$ & $-1+X^2$ & $X^2$ & $-1$ \\
\hline
\end{tabular}
\end{center}
\caption{The elements of $\mathbb{Z}/17\mathbb{Z}$ in $\mathcal{B} = (17, 3, 7, 2)$}
\label{tab:exple1}
\end{table}
\end{example}
An Adapted Modular Number System (AMNS) is a MNS with a special property:  ~$\gamma$ must be a root of the polynomial $E(X)=X^n - \alpha X - \lambda$ in $({\mathbb{Z}/p\mathbb{Z}})[X]$, where $\alpha$ and $\lambda$ are ``small'' integers. Notice that such  a property cannot be always satisfied for any tuple $(p,\alpha,\lambda)$. When such a tuple exists, this property allows to speed up arithmetic operations (see \cite{thPl_05,pmns_05} for more details). 

There are two strategies to build an AMNS: either by choosing the size of the modulus $p$, or by choosing the value of $p$. 
The first strategy has been studied by Bajard and al. in \cite{amns_04,thPl_05}. It leads to representation systems that allow very efficient modular arithmetic. But these AMNS are irrelevant for some cryptographic standards where the value of $p$ is already fixed.
With the second strategy, parameters $(n,\gamma,\rho)$ of the AMNS are built from a fixed value $p$. For the same value $p$, it is even possible to generate many distinct AMNS, which could be an interesting property to build side-channel resistant countermeasures. The main drawback of this strategy is that its parameters generation is more complex and also its leads to AMNS with arithmetic operations generally less efficient than those of the first strategy. 

Whether the value or the size of $p$ is fixed, arithmetic operations (addition, multiplication) in the corresponding AMNS remain the same. 
Only one operation differs:~the \textit{internal reduction}.
The internal reduction is a method that takes as input an element $X = (x_0, \dots, x_{n-1})$
and outputs $Y = (y_0, \dots, y_{n-1})$ such that 
$$\sum_{i=0}^{n-1} y_i\gamma^i\equiv \sum_{i=0}^{n-1} x_i\gamma^i\ (\text{mod}\  p),$$
with $|y_i|< \rho$, i.e. $Y$ is a representative of $X$ in $\mathcal{B}$.
As soon as a product (or an addition) of two elements in the AMNS is done, the internal reduction process may be required to guarantee that the result still be in the AMNS.

When only the size of $p$ is fixed, the internal reduction is essentially a vector matrix multiplication. This operation is very cheap because the generation process allows to choose a very sparse matrix (see \cite{amns_04,thPl_05}) with non zero elements equal to $\pm$ powers of two.

When the value of $p$ is fixed, there is no such freedom, so other methods have been proposed. In \cite{pmns_05}, Bajard and al. described two methods to perform the internal reduction: one using lookup tables and another one using a Barrett-like approach. In \cite{amns_08}, using a trick similar to Montgomery's modular multiplication, N\`egre and al. described another way to compute the internal reduction. This method is currently the best known when $p$ is fixed. The main difficulty is that it requires a polynomial $M'$ (see subsection \ref{subsec:genparam}) which existence is not easy to ensure. The authors gave in \cite{amns_08} a condition of existence which is unfortunately not sufficient. In \cite{amns_12}, El Mrabet and al. choosed the polynomial $E$ such that $\alpha=0$, $\lambda=-1$ and showed that in this context, the polynomial $M'$ always exists. Unfortunately, their proof has some issues.

In this paper, we consider the generation of an AMNS when the value of $p$ is already known. We prove that for the polynomial $E(X)=X^n-\lambda$ (i.e. $\alpha=0$), the polynomial $M'$ always exists for any $\lambda$ and we explain to obtain it. We describe a complete set of algorithms for arithmetic operations using the Montgomery-like method for internal reduction and we show how to generate all the parameters needed for these algorithms.

When computing the sum of two elements in the AMNS, a polynomial addition is done. Then an internal reduction may be required to keep the result in the AMNS. Regardless how this reduction is done, its cost could be too expensive compared to that of a simple polynomial addition. In this paper, we detail an idea that allows to build, from  a parameter $\Delta$, an AMNS such that at most $\Delta$ additions followed by a multiplication and only one reduction can be computed with the corresponding result still in the AMNS. In the last part, we talk about the implementation, give some examples of AMNS and show some comparisons with popular big integer libraries. 



\section{Arithmetic operations in AMNS}

Let $\mathcal{B} = (p,n,\gamma,\rho,E)$ be an AMNS, $a\in{\mathbb{Z}/p\mathbb{Z}}$ and $A$ its representative in $\mathcal{B}$. 
Most of the time, an element $A$ will be considered as a polynomial $A(X)$ of degree $n-1$ such that $A(X) = a_0 + a_1.X + \dots + a_{n-1}.X^{n-1}$ and $A(\gamma) \equiv a \; ({\rm mod} \, p)$.  

As already mentionned in the introduction, the general form of the polynomial $E$ is $E(X) = X^n - \alpha X - \lambda$, with $E(\gamma) \equiv 0 \, ({\rm mod} \, p)$.
$E$ is called the \textit{external reduction polynomial} and it is used to reduce polynomials which degree exceeds $n-1$. In this paper, we consider the case $\alpha = 0$ (i.e $E(X) = X^n - \lambda$). This choice makes arithmetic operations even faster. Moreover in this context, finding a root $\gamma$ (modulo $p$) of $E(X)$ boils down to compute a nth-root of $\lambda$ modulo $p$ (see \cite{nthroot}). 
\begin{proposition}
\label{prop:rho-min}
Let $\mathcal{B} = (p,n,\gamma,\rho,E)$ be an AMNS. The integer $p$ satisfies
 $p \leqslant (2\rho)^n$, hence $\lceil \sqrt[n]{p}/2 \rceil \leqslant \rho$.
\end{proposition}
\begin{proof}
The number of elements in $\mathcal{B}$ is $(2\rho)^n$, as elements can have negative coefficients and their absolute values are bounded by $\rho$ (see definition \ref{def:mns}). We want to represent all elements in $\mathbb{Z}/p\mathbb{Z}$, so $\rho$ must be such that $p \leqslant (2\rho)^n$.
\end{proof}
For cryptographic needs, the most important operations are addition and multiplication. But, as elements representation in AMNS is different from the binary one, we also need conversion methods to find a representative of each integer modulo $p$ in the AMNS (and vice-versa). 
All these operations (conversion, addition, multiplication) will use an internal reduction process.  Let $\mathcal{B} = (p, n, \gamma, \rho, E)$ be an AMNS such that $E(X) = X^n - \lambda$ and $\gamma^n \equiv \lambda \, \text{mod} \, p$, the internal reduction process  maps a polynomial $V(X)\in\mathbb{Z}[X]$ to a polynomial $\tilde V(X)\in\mathbb{Z}[X]$ such that $\tilde V(\gamma)\equiv V(\gamma)\pmod p$, $\|\tilde V \|_{\infty} < \|V \|_\infty$ and $\deg \tilde V = \deg V$.

\subsection{The internal reduction}
In \cite{pmns_05}, the authors suggest to use  a Barrett-like algorithm for the reduction procedure. In \cite{amns_08}, authors improve the reduction procedure by using a Mongtomery-like algorithm (see Algorithm \ref{alg:redcoeff}).
\begin{algorithm}[H]
  \caption{RedCoeff($V\in\mathbb{Z}[X]$) (Coefficient reduction)}
  \label{alg:redcoeff}
  \begin{algorithmic}[1]
    \Require $deg(V) \leqslant n-1$, $\mathcal{B} = (p, \, n, \, \gamma, \, \rho, E)$, $M \in \mathcal{B}$, such that $M(\gamma)\equiv 0\pmod p$, an integer $\phi$ and $M'\equiv -M^{-1}\ \text{mod}(E,\phi)$.
    \vspace{1mm}
    \Ensure $S(\gamma) = V(\gamma)\phi^{-1} \, \text{mod} \, p$
    \vspace{1mm}
    \State $Q \leftarrow V \times M' \, \text{mod} \, (E, \, \phi)$
    \vspace{1mm}
    \State $R \leftarrow (V + Q \times M \, \text{mod}\, E)$
    \vspace{1mm}
	\State $S \leftarrow R / \phi$
    \vspace{1mm}
    \State return $S$
  \end{algorithmic}
\end{algorithm}
\noindent It is proved in \cite{amns_08} that:
\begin{proposition}
\label{prop-amns-08}
Let $\sigma = \|M\|_{\infty}$.
If $V$, $\rho$ and $\sigma$ are such that:
\begin{center}
 $\|V\|_{\infty} \leqslant n|\lambda|\rho^2$, \;$\rho \geqslant 2|\lambda|n\sigma$ \; and \; $\phi \geqslant 2|\lambda|n\rho$
\end{center} 
then $S$ (the output of the algorithm \ref{alg:redcoeff}) is such that $\|S\|_{\infty} < \rho$ (i.e $S \in \mathcal{B})$.
\end{proposition}
This property is useful for multiplication. If $A$ and $B$ are two elements of $\mathcal{B}$, then the coefficients of $V(X) = A(X)B(X)\pmod {E(X)}$ meet the requirements of pro\-position \ref{prop-amns-08} as proved in  \cite{amns_08}. Hence, one single call to RedCoeff will output an element $S\in\mathcal B$, such that  $S(\gamma)\equiv A(\gamma)B(\gamma)/\phi\pmod p$, similarly to Montgomery modular reduction method. 

Nevertheless, using this algorithm, the computation of the product $\alpha_1\times\alpha_2\times\cdots\times\alpha_k$, with $\alpha_i\in\mathbb{Z}/p\mathbb{Z}$, will output a polynomial $S$ such that
$S(\gamma)\equiv\phi^{-k}\prod_{i=1}^k\alpha_i\pmod p$. In order to ensure that the operations in AMNS are consistent, we propose to use the same idea as the Montgomery modular reduction method:~  any element $a\in\mathbb{Z}/p\mathbb{Z}$  will be replaced by  $a.\phi\pmod p$. 
Thus, any element $a\in\mathbb{Z}/p\mathbb{Z}$ will have a representative $A(X)\in\mathcal{B}$ such that $A(\gamma)\equiv a\phi\pmod p$. As a consequence, if  $A(X)$, $B(X)$ $\in\mathcal{B}$  respectively represent $a$, $b$ $\in\mathbb{Z}/p\mathbb{Z}$, then the reduction procedure applied to $V(X) = A(X)B(X)\pmod {E(X)}$ will output  a polynomial $S\in\mathcal B$, such that  $S(\gamma)\equiv ab\phi\pmod p$. We can now detail the conversion operations.
\begin{remark}
\label{rem:smallM}
The algorithm \ref{alg:redcoeff} requires the existence of a polynomial $M$ which is invertible mod $(E,\phi)$. We will discuss this important point in subsection \ref{subsec:genparam}.
Moreover, proposition \ref{prop-amns-08} implies that $\rho > \|M\|_{\infty}$, hence we want $\|M\|_{\infty}$ as small as possible.\\
\end{remark}
In the sequel of the paper, we will assume that $\rho \geqslant 2|\lambda|n\|M\|_{\infty}$ and $\phi \geqslant 2|\lambda|n\rho$. 
\subsection{Conversion operations}
\subsubsection{Conversion from binary representation to AMNS.}
We present here two methods that can be used for conversion from classical representation to AMNS. The first method (inspired from \cite{pmns_05}) involves precomputations and needs some storage capacity while the other one is less memory consuming but needs more computations.
\subsubsection*{Method 1} 
First {precompute} representatives $P_i(X)$ of $\rho$ powers in $\mathcal{B}$, i.e $P_i \equiv (\rho^i)_\mathcal{B}$, for $i=1, \, \dots, n-1$ (see section\ref{subsec:computRho}). The conversion is then obtained using the algorithm \ref{alg:conv_to_amns}. 
\begin{algorithm}[H]
  \caption{Conversion from classical representation to AMNS}
  \label{alg:conv_to_amns}
  \begin{algorithmic}[1]
    \Require $a \in \mathbb{Z}/p\mathbb{Z}$ \, and \, $\mathcal{B} = (p, \, n, \, \gamma, \, \rho, E)$
    \vspace{1mm}
    \Ensure $A \equiv (a.\phi)_\mathcal{B}$
    \vspace{1mm}
    \State $b = (a.\phi^2) \, \text{mod}\, p$
    \vspace{1mm}
    \State $b = (b_{n-1}, ..., b_0)_\rho$
	\vspace{1mm}    
    \State $U \leftarrow \sum\limits_{i=0}^{n-1} b_i.P_i(X)$
    \vspace{1mm}
    \State $A \leftarrow \text{RedCoeff}(U)$
    \vspace{1mm}
    \State return $A$
  \end{algorithmic}
\end{algorithm}
At line 3, $U$ is a representative of $a.\phi^2$ and $\|U\|_{\infty} < n.\rho^2$. Hence from proposition \ref{prop-amns-08}, $\|A\|_{\infty} < \rho$ and $A \equiv (a.\phi)_\mathcal{B}$.
\subsubsection*{Method 2}
\label{remark:output-redcoeff}
RedCoeff procedure outputs $S$  such that:\\
$\|S\|_{\infty} < (\|V\|_{\infty} + n\phi|\lambda|\sigma)/\phi$.
If $\rho \geqslant 2n|\lambda|\sigma$ and\break $\phi \geqslant 2|\lambda|n\rho$, this implies that:
\begin{center}
$\|S\|_{\infty} < \frac{\|V\|_{\infty}}{2|\lambda|n\rho} + \frac{\rho}{2} \leqslant \frac{\|V\|_{\infty}}{2\rho} + \frac{\rho}{2}\,.$
\end{center}
So, if \, $\|V\|_{\infty} < \rho^2$, then $\|S\|_{\infty} < \rho$ (i.e: $S \in \mathcal{B}$). 
But if \, $\|V\|_{\infty} \geqslant \rho^2$, then $\|S\|_{\infty} < \|V\|_{\infty}/\rho$. That is, one  call of RedCoeff divides the coefficients of $V$ by at least $\rho$ if $\|V\|_{\infty} \geqslant \rho^2$.
We propose a conversion procedure based on this fact.\\
\noindent\\
Let $T = \phi^n \, \text{mod} \, (p)$.  The algorithm \ref{alg:conv_to_amns2} describes another way to compute the conversion of an element of $\mathbb{Z}/p \mathbb{Z}$.
\begin{algorithm}[H]
  \caption{Conversion from binary representation to AMNS}
  \label{alg:conv_to_amns2}
  \begin{algorithmic}[1]
    \Require $a \in \mathbb{Z}/p\mathbb{Z}$, $\mathcal{B} = (p, \, n, \, \gamma, \, \rho, E)$ and $T = \phi^n \, \text{mod} \, (p)$
    \vspace{1mm}
    \Ensure $A \equiv (a.\phi)_\mathcal{B}$
    \vspace{1mm}
    \State $A = (a.T) \, \text{mod}\, p$  /* polynomial of degree 0 */
    \vspace{1mm}
    \For{$i=1 \dots n-1$}
    \vspace{1mm}
    \State $A \leftarrow \text{RedCoeff}(A)$
    \vspace{1mm}
    \EndFor
    \vspace{1mm}
    \State return $A$
  \end{algorithmic}
\end{algorithm}
As previously mentionned one single call to RedCoeff method divides its input coefficients by at least $\rho$. At line 1, $A$ is a polynomial of degree 0 which constant coefficient is strictly less than $p$ $(\leqslant \rho^n)$. As a consequence, calling $n-1$ times RedCoeff on $A$  ensures that the algorithm will output $A $ in $\mathcal{B}$ with $A \equiv (a.\phi)_\mathcal{B}$.
This method is slower than the first method , but its advantage is that only the precomputation of $T$ is needed. 
\begin{remark}
This method always works because $\phi^n \geqslant p$ (as $(2.\rho)^n \geqslant p$ and $\phi \geqslant 2.\rho$), so $\phi^n \neq T$. 
Otherwise, its output $A$ could be such that $\|A\|_{\infty} \geqslant \rho$, and so $A$ will not be in $\mathcal{B}$.
\end{remark}
\subsubsection{Conversion from AMNS to binary representation.}
We present here a slight modification of two methods given in \cite{pmns_05} for computing the integer value corresponding to an element in an AMNS $\mathcal{B}$. 
\subsubsection*{Method 1}
The integer represented by $A$ in $\mathcal{B}$ is the value $A(\gamma)\phi^{-1}\pmod p$. It can be computed using the classical Horner's method as described in Algorithm \ref{alg:conv_from_amns}.
\begin{algorithm}[h]
  \caption{Conversion from AMNS to classical representation}
  \label{alg:conv_from_amns}
  \begin{algorithmic}[1]
    \Require $A \in \mathcal{B}$ and $\mathcal{B} = (p, \, n, \, \gamma, \, \rho, E)$
    \vspace{1mm}
    \Ensure $a = A(\gamma).\phi^{-1} \, \text{mod}\, p$
    \vspace{1mm}
    \State $A \leftarrow \text{RedCoeff}(A)$
	\vspace{1mm} 
	\State $a \leftarrow 0$
	\vspace{1mm}     
    \For{$i=n-1 \dots 0$}
    \vspace{1mm}
    \State $a \leftarrow (a\gamma + A_i) \, \text{mod}\, p$
    \vspace{1mm}
    \EndFor
    \vspace{1mm}
    \State return $a$
  \end{algorithmic}
\end{algorithm}
\subsubsection*{Method 2} 
This method (Alg. \ref{alg:conv_from_amns2}) is a simple evaluation of a polynomial with {precomputed} powers  $\gamma$ modulo $p$. We note them $g_i = \gamma^i \pmod p$, for $i=0, \, \dots, n-1$. 
\begin{algorithm}[h]
  \caption{Conversion from AMNS to classical representation}
  \label{alg:conv_from_amns2}
  \begin{algorithmic}[1]
    \Require $A \in \mathcal{B}$ and $\mathcal{B} = (p, \, n, \, \gamma, \, \rho, E)$
    \vspace{1mm}
    \Ensure $a = A(\gamma).\phi^{-1} \, \text{mod}\, p$
	\vspace{1mm} 
    \State $A \leftarrow \text{RedCoeff}(A)$
	\vspace{1mm} 
	\State $a \leftarrow 0$
	\vspace{1mm}     
    \For{$i=n-1 \dots 0$}
    \vspace{1mm}
    \State $a \leftarrow a + A_ig_i$
    \vspace{1mm}
    \EndFor
	\vspace{1mm} 
	\State $a \leftarrow a\ \text{mod}\ p$    
   \State return $a $
  \end{algorithmic}
\end{algorithm}
\begin{remark}
The two algorithms start by a call to RedCoeff in order to compute a polynomial $B$ such that $B(\gamma)\equiv A(\gamma)\phi^{-1}\pmod p$. Such a call is expected to be less expensive than a direct computation of $A(\gamma)\phi^{-1}\pmod p$. In Algorithm \ref{alg:conv_from_amns}, at each step of the loop, the product of two elements of size $\log_2(p)$ is computed, and a modular reduction is done. In Algorithm  \ref{alg:conv_from_amns2}, at each step of the loop, the product of one element of size  $\log_2(p)$ with an element of size $\log_2 (\rho)$ is computed.
Hence the size of the last computed value is about $\log_2(n)+\log_2(\rho)+\log_2(p)$. Only one modular reduction is computed to obtain the result. This method requires the precomputation of some powers of $\gamma$ modulo $p$, but is much faster than Algorithm \ref{alg:conv_from_amns}.
\end{remark}
\subsection{Multiplication}
\begin{algorithm}[H]
  \caption{Multiplication in AMNS}
  \label{alg:amns_mult}
  \begin{algorithmic}[1]
    \Require $A \in \mathcal{B}$, $ B \in \mathcal{B}$ and $\mathcal{B} = (p, \, n, \, \gamma, \, \rho, E)$
    \vspace{1mm}
    \Ensure $S \in \mathcal{B}$ with $S(\gamma) \equiv A(\gamma).B(\gamma).\phi^{-1} \, \text{mod}\, p$
    \vspace{1mm}
    \State $V \leftarrow A.B \, \text{mod}\, E$
	\vspace{1mm}    
    \State $S \leftarrow \text{RedCoeff}(V)$
    \vspace{1mm}
    \State return $S$
  \end{algorithmic}
\end{algorithm} 
The multiplication of two polynomials increases the degree of the result. As a consequence, we must first reduce the degree and next reduce coefficient sizes. The reduction of the degree is done using the external reduction polynomial $E$.
As $E(\gamma) \equiv 0 \, (\text{mod}\, p)$, then $A.B(\gamma) \equiv V (\gamma) \,(\text{mod}\, p)$. 
After the external reduction, the result $V$ has a degree lower than $n-1$ but we only have $\|V\|_{\infty} \leqslant n|\lambda|\rho^2$. So, the next and final step is to reduce the sizes of its coefficients with RedCoeff method so that theirs absolute values are lower than $\rho$.

\subsection{Addition}
\label{subsec:addition}
\begin{algorithm}[H]
  \caption{Addition in AMNS}
  \label{alg:amns_add}
  \begin{algorithmic}[1]
    \Require $A \in \mathcal{B}$, $ B \in \mathcal{B}$ and $\mathcal{B} = (p, \, n, \, \gamma, \, \rho, E)$
    \vspace{1mm}
    \Ensure $S = A + B$ 
    \vspace{1mm}
    \State $S \leftarrow A + B$
    \vspace{1mm}
    \State return $S$
  \end{algorithmic}
\end{algorithm}

The addition is simply a sum of two polynomials which doesn't increase the result degree. However, the result $S$ may not satisfy $\|S\|_{\infty} < \rho$. The simplest way to have the output $S$ always in $\mathcal{B}$ is, if $\|S\|_{\infty} \geqslant \rho$, to perform a coefficient reduction on $S$ followed by a multipcation by a representative of $\phi^2$, using algorithm \ref{alg:amns_mult}. Unfortunaly, this is too expensive for an addition. \\

We propose an alternative method.
Suppose that, for the target application, one needs to compute at most $\Delta$ consecutive additions so that the final output $S$ is such that $\|S\|_{\infty} \leqslant (\Delta + 1)\rho$. Suppose now that 
the product $S\times T$ must be computed where $T$ is the output of $\Theta$  consecutive additions ($0 \leqslant \Theta \leqslant \Delta$). Let $S$ and $T$ be two such elements (i.e $\|S\|_{\infty} \leqslant (\Delta + 1).\rho$ and $\|T\|_{\infty} \leqslant (\Theta + 1).\rho$). Using algorithm \ref{alg:amns_mult} to multiply $S$ and $T$, one has the output $W$ such that:
\begin{center}
$\|W\|_{\infty} < \frac{n|\lambda|(\Delta + 1)(\Theta+1)\rho^2}{\phi} + \frac{\rho}{2}\,.$
\end{center}
So, as we want $\|W\|_{\infty} < \rho$, it suffices to take $\phi$ such that: \begin{center}
$\phi \geqslant 2(\Delta + 1)^2|\lambda|n\rho\,.$
\end{center}
\textbf{Note}: For $\Delta = 0$, we obtain $\phi \geqslant 2|\lambda|n\rho$, which corresponds to what was suggested by proposition \ref{prop-amns-08}. \\

Until now, we have defined the AMNS and shown how to perform essential arithmetic operations in it. In the following, we will show how to generate all
necessary parameters.

\section{AMNS parameters generation}
In this section, we describe how to generate all necessary parameters of an AMNS. We assume that $p$ is a prime number larger than 3,
thus $p-1$ is always even. We  also assume that $\lambda  < p$, for consistency.\\

The complete set of parameters of $\mathcal{B}$ is:
\begin{itemize}
	\item $p$: a prime integer
	\item $n$: the number of coefficients of elements in the AMNS
	\item $\lambda$: a ``small'' integer.
	\item $\gamma$: a nth-root of $\lambda$ modulo $p$
	\item $\rho$: the upper-bound on the $\|.\|_\infty$ of the elements of  $\mathcal{B}$.
	\item $\phi$: the integer used in RedCoeff.
	\item $E$: the external reduction polynomial ($E(X) = X^n - \lambda$)
	\item $M$: the internal reduction polynomial.
	\item $M'$: a polynomial such that $M' = - M^{-1} \, \text{mod}\, (E, \, \phi)$
	\item $\Delta$: the maximum number of consecutive additions that can be done before a multiplication (see subsection \ref{subsec:addition}).
\end{itemize}

Amongst these parameters, $n$ and $p$ are choosen with regard to the target  application and the target architecture. The parameters 
$\lambda$ and $E$ depend on the existence of $\gamma$ (section \ref{subsec:existence}). The proof on the existence of the polynomials $M$ and $M'$ is detailled in section \ref{subsec:genparam} and the generation of $M$ in section \ref{subsec:genM}. Finally, the computation of the representatives of powers of $\rho$ that are used in the algorithm \ref{alg:conv_to_amns} is described in section \ref{subsec:computRho}.

\subsection{Existence of $\mathit{\gamma}$}
\label{subsec:existence}
The first constraint in the generation of an AMNS is the existence and the computation of $\gamma$, a nth-root modulo $p$ of $\lambda$.
\begin{proposition}
	\label{prop-gmm}
	Let $E(X) = X^n - \lambda$, for $\lambda \in \mathbb{Z}\setminus\{0\}$. 
	Let $g$ be a generator of $(\mathbb{Z}/p\mathbb{Z})^\times$ and $y$ such that $g^y \equiv \lambda \, \text{\rm mod}\, p$. If $gcd(n, p-1)\mid y$, then there exists $gcd(n, p-1)$ roots $\gamma$ of $E(X)$ in $\mathbb{Z}/p\mathbb{Z}$.
\end{proposition}
\begin{proof}
	If $gcd(n, p-1)\mid y$, the equation $nx\equiv y\pmod{p-1}$ has $k$ solutions $x$, where $k=gcd(n,p-1)$. Let $x_0$ be one of these solutions, and let consider $\gamma\equiv g^{x_0}\pmod{p-1}$. Then $\gamma^n\equiv\lambda\pmod p$.
\end{proof}
Proposition \ref{prop-gmm} gives a condition that guarantees the existence of $\lambda$ nth-roots modulo $p$ and their number. But it requieres to compute $y$ such that $g^y \equiv \lambda \, \text{\rm mod}\, p$ (i.e the discrete logarithm of $\lambda$ in  $\mathbb{Z}/p\mathbb{Z}$) which can be very hard if $p$ is big enough. So, we give bellow sufficient (but not necessary) conditions that are easier to satisfy  and that guarantee the existence of a nth-root modulo $p$ of $\lambda$, taking eventually $\lambda = 1$.
\begin{corollary}
	\label{coro:lmbd}
	If $gcd(n, p-1)=1$ then there exists a unique nth-root $\gamma$ of $\lambda$ in $\mathbb{Z}/p\mathbb{Z}$, for any $\lambda \in \mathbb{Z}\backslash\{0\}$.
\end{corollary}
\begin{proof}
	If $gcd(n, p-1)=1$, this nth-root can be easily computed. In fact, using the extended euclidean algorithm, we compute Bezout coefficients for $(n, p-1)$. That is, $u$ and $v$ in $\mathbb{Z}$ such that $n.u + (p-1).v = 1$. So, $\lambda = \lambda^{n.u + (p-1).v} = (\lambda^u)^n.(\lambda^{p-1})^v$. As, $\lambda^{p-1} \equiv 1 \pmod p$, it is obvious that $\lambda \equiv (\lambda^u)^n \pmod p$. Which means that $\lambda^u \pmod p$ is a nth-root modulo $p$ of $\lambda$.
\end{proof}
If $gcd(n, p-1)=1$ and $\lambda = 1$, then, using corollary \ref{coro:lmbd}, the unique nth-root $\gamma$ of $\lambda$ is $1$.
With $\lambda = 1$, an AMNS can not be build because the maximum value that can be computed is lower than $n\rho$, with $\rho \approx \sqrt[n]{p}$. So, it will not be posssible to generate all elements in $\mathbb{Z}/p\mathbb{Z}$. Below, we give a corollary for the case $\lambda = 1$.

\begin{corollary}
	\label{coro:lmbd1}
	If $gcd(n, p-1) > 1$ then there exists at least one non-trivial nth-root $\gamma$ of $1$. So, one can take $\lambda = 1$.
\end{corollary}
\begin{proof}
	Let $\lambda = 1$, we are looking for $\gamma$ such that $\gamma^n \equiv 1 \pmod p$.
	It is well known that there are $gcd(n, p-1)$ nth-roots of unity modulo $p$. Let $g$ be a generator of $(\mathbb{Z}/p\mathbb{Z})^\times$, $d = gcd(n,p-1)$ with $d>1$, and  let $h = g^{(p-1)/d} \pmod p$. Then, $h^n \equiv 1 \pmod p$ and $h \neq 1$. So, $h$ is a non-trivial nth-root of $\lambda$. \\
	The other nth-roots are $h^i \pmod p$, for $2 \leqslant i \leqslant d$.
\end{proof}
\begin{remark}
	If $gcd(n, p-1) > 1$, it is not necessary to take $\lambda = 1$. There exists efficient algorithms (like in \cite{nthroot}) that can be used to compute the nth-root of an element $\lambda$ if such a root exists.
\end{remark}

\subsection{Existence of $M$ and $M'$}
\label{subsec:genparam}

Section \ref{subsec:existence} shows that once $n$ is fixed, the nth-root $\gamma$ of $\lambda$ can be easily obtained (taking eventually $\lambda=1$).
To represent all elements in $\mathbb{Z}/p\mathbb{Z}$, $\rho$ must be such that $\lceil\sqrt[n]{p}/2\rceil \leqslant \rho$ (see proposition \ref{prop:rho-min}). \\
For internal reduction, we also want $\rho \geqslant 2|\lambda|n\|M\|_{\infty}$ and $\phi \geqslant 2(\Delta + 1)^2|\lambda|n\rho$ (as discussed in subsection \ref{subsec:addition}). 

Now, it remains to see how to generate $M$ and $M'$. In \cite{amns_08}, the authors state that $M$ must be chosen such that gcd($E$,$M$)=1, but this does not guarantee the existence of $M'$.
Indeed, if gcd($E$,$M$)=1 then it exists $M'\in\mathbb{Q}[X]$ such that $MM'\equiv 1\pmod E$, but nothing guarantee that the coefficients of $M'$ are invertible modulo $\phi$. Hence, there is no evidence that gcd($E$,$M$)=1 implies that $M'M\equiv 1\pmod{(E,\phi)}$.

In section 3.3 of \cite{amns_09}, when the value of $p$ is already fixed and $\phi$ is a power of 2, the authors show how to build a lattice which reduced basis always contains a polynomial $M$ that is invertible modulo $(E,\phi)$. Unfortunately, their proof uses the fact that a polynomial $M$ is invertible modulo $(E,\phi)$ as soon as the evaluation of $M$ over all integers is odd. This is a necessary condition but not a sufficient one. As an example, let $E(X)=X^6+1$ and $M(X)=X^4+X^2+1$, the evaluation of $M$ is odd over all integers but it is not invertible modulo $(E,\phi)$ for any even value of $\phi$ because the resultant of $E$ and $M$ is $16$. This leads us to first recall some essential elements about the resultant of two polynomials. We will then give a proof of the existence of $M$ and $M'$.
\begin{definition}[Resultant]
	Let $\mathcal{A}$ be an integral domain. Let $A$ and $B$ be two polynomials in $\mathcal{A}[X]$. The resultant $Res(A,B)$ of $A$ and $B$ is the determinant of their Sylvester matrix. Therefore, it is an element of $\mathcal{A}$. 
\end{definition}

If $A(X) = a_0 + a_1X + \dots + a_{n}X^{n}$ and $B(X) = b_0 + b_1X + \dots + b_{m}X^{m}$, then their Sylvester matrix is defined as follows: 
\begin{center}
	\begin{small}
		$
		\begin{pmatrix} 
		a_{n} & 0 & \dots & 0 & b_{m} & 0 & \dots & 0 \\
		a_{n-1} & a_{n} & \ddots & \vdots & \vdots & b_m & \ddots & \vdots \\
		\vdots & a_{n-1} & \ddots & 0 & \vdots & & \ddots & 0 \\
		\vdots & \vdots & \ddots & a_n & b_1 & & & b_m \\
		a_0 & & & a_{n-1} & b_0 & \ddots & \vdots & \vdots \\
		0 & \ddots & & \vdots & 0 & \ddots & b_1 & \vdots \\
		\vdots & \ddots & a_0 & \vdots & \vdots & \ddots & b_0 & b_1 \\
		0 & \dots & 0 & a_0 & 0 & \dots & 0 & b_0 \\
		\end{pmatrix}$
	\end{small}
\end{center}
In $\mathbb{Z}[X]$, there is no Bezout's identity, but the following essential property will help us to set an existence criteria for the polynomial $M'(X)$.
\begin{proposition}
	\label{prop:essential}
	Let $\mathcal{A}$ be an integral domain and let  $A$ and $B$ two non-zero polynomials in $\mathcal{A}[X]$ such that $\deg(A)+\deg(B) \geqslant 1$. There exist  $U$ and $V$ in $\mathcal{A}[X]$ such that   $A(X)U(X) + B(X)V(X) = Res(A,B)$, $\deg (U) < \deg (B)$, and $\deg (V) < \deg(A)$.
\end{proposition}
We can now state our existence criteria for the polynomial $M'$.
\begin{proposition}[Existence criteria]
	Let $M\in\mathbb{Z}[X]$, $E\in\mathbb{Z}[X]$ and $\phi \geqslant 2$ an integer. 
	
	If $\text{gcd}(\text{Res}(E,M),\phi)=1$ then there exists $M' \in\mathbb{Z}[X]$ such that $M'M\equiv 1\ \text{mod}\  (E,\phi)\,.$ 
\end{proposition}
\begin{proof}
	Let $r=$ Res($E$,$M$) and let $\psi$ the inverse of $r$ modulo $\phi$. From the proposition \ref{prop:essential}, there exist  $U$ and $V$ in $\mathbb{Z}[X]$ such that   $U(X)M(X) + V(X)E(X) = r$. Hence $\psi U(X)M(X)+\psi V(X)E(X)=\psi r$, which implies $\psi U(X)M(X)\equiv 1\ \text{mod}\ ({E,\phi})$. 
\end{proof}
\begin{remark}
	In RedCoeff method (alg. \ref{alg:redcoeff}), an exact division by $\phi$ is done. Because we want this method to be fast, we usually take $\phi$ as a power of two. So for $Res(E,M)$ to be invertible in $\mathbb{Z}/\phi\mathbb{Z}$, it must be odd, as invertible elements in $\mathbb{Z}/\phi\mathbb{Z}$ are those which are prime with $\phi$.
\end{remark}
As mentioned in remark \ref{rem:smallM}, $M$ must be chosen so that  $\|M\|_{\infty}$ is as small as possible. In the sequel of this paper, we will assume that $\phi$ is a power of two and we are going to show how to generate $M$ such that $Res(E,M)$ be odd and  $\|M\|_{\infty}$ be  small.
\subsection{Generation of the polynomial $M$}
\label{subsec:genM}
Let $\mathcal{L}$ be the lattice of all polynomials having $\gamma$ as root modulo $p$ and whose degree is at most $n-1$:
\begin{center}
	$\mathcal{L} = \{a(X) \in \mathbb{Z}[X]$, such that: $deg(a) < n$ and $a(\gamma) \equiv 0 \, \text{mod}\, (p)\}$
\end{center}
The idea for finding $M$ is to compute a reduced basis of $\mathcal{L}$ using a lattice reduction algorithm (like LLL algorithm for example) and then to take $M$ as an element (or a special combination of elements) of this reduced basis. In order to find a suitable polynomial $M$, we will distinguish two cases according to the parity of $\lambda$. We recall that for practical reason, $\phi$ is  a power of two. In the sequel, we will use the  following property:
\begin{property} 
	\label{rmk:mat_parity}
	\textbf{\textit{(Matrix determinant parity)}}
	\noindent\\
	Let $A, B \in \mathcal{M}_{n \times n}$ be two matrices with elements in $\mathbb{Z}$, such that: $A = (a_{ij})_{0 \leqslant i,j <n}$ \, and \, $B = (b_{ij})_{0 \leqslant i,j <n}$ with $b_{ij} = a_{ij}\pmod 2$. Then, the determinants of $A$ and $B$ have the same parity.
\end{property}
Let  $E(X) = X^n - \lambda$ and  $M(X) = m_0 + m_1X + \dots + m_{n-1}X^{n-1}$. The Sylvester matrix of $E$ and $M$ is defined as follows: \begin{center}
	\begin{small}
		$\mathcal{S}_{M,E} =
		\begin{pmatrix} 
		1 & 0 & \dots & 0 & m_{n-1} & 0 & \dots & 0 & 0\\
		0 & 1 & \dots & 0 & m_{n-2} & m_{n-1} & \dots & 0 & 0 \\
		\vdots & & \ddots & & \vdots & & & & \vdots \\
		0 & 0 & \dots & 1 & m_1 & m_2 & \dots & m_{n-1} & 0 \\
		0 & 0 & \dots & 0 & m_0 & m_1 & \dots & m_{n-2} & m_{n-1} \\
		-\lambda & 0 & \dots & 0 & 0 & m_0 & \dots & m_{n-3} & m_{n-2} \\
		\vdots & & & & \vdots & & \ddots & & \vdots \\
		0 & 0 & \dots & 0 & 0 & 0 & \dots & m_0 & m_1 \\
		0 & 0 & \dots & -\lambda & 0 & 0 & \dots & 0 & m_0 \\
		\end{pmatrix}$
	\end{small}
\end{center}
\subsubsection*{\textit{{Case 1: $\lambda$ is even}}}
If $\lambda$ is even, there is a very simple condition on $M$ to guarantee that Res($E$,$M$) = det($\mathcal{S}_{M,E}$) be odd.
\begin{proposition}
	\label{prop:even-m-ex}
	Let $E(X) = X^n - \lambda$ such that $\lambda$ is even. Let $M = m_0 + m_1X + \dots + m_{n-1}X^{n-1}$ be a polynomial such that $M \in \mathcal{L}$. Then, Res($E$,$M$) is odd if and only if $m_0$ is odd.
\end{proposition}
\begin{proof}
	If $\lambda$ is even, then using property \ref{rmk:mat_parity}, it is obvious that the determinant of $\mathcal{S}_{M,E}$ has the same parity as that of the following matrix, where $\overline{m_i} = m_i\pmod 2$: 
	\begin{center}
		\begin{small}
			$
			\begin{pmatrix} 
			1 & 0 & \dots & 0 & \overline{m_{n-1}} & 0 & \dots & 0 & 0\\
			0 & 1 & \dots & 0 & \overline{m_{n-2}} & \overline{m_{n-1}} & \dots & 0 & 0 \\
			\vdots & & \ddots & & \vdots & & & & \vdots \\
			0 & 0 & \dots & 1 & \overline{m_1} & \overline{m_2} & \dots & \overline{m_{n-1}} & 0 \\
			0 & 0 & \dots & 0 & \overline{m_0} & \overline{m_1} & \dots & \overline{m_{n-2}} & \overline{m_{n-1}} \\
			0 & 0 & \dots & 0 & 0 & \overline{m_0} & \dots & \overline{m_{n-3}} & \overline{m_{n-2}} \\
			\vdots & & & & \vdots & & \ddots & & \vdots \\
			0 & 0 & \dots & 0 & 0 & 0 & \dots & \overline{m_0} & \overline{m_1} \\
			0 & 0 & \dots & 0 & 0 & 0 & \dots & 0 & \overline{m_0} \\
			\end{pmatrix}$ 
		\end{small}
	\end{center}
	It is an upper triangular matrix with only the value $1$ or $\overline{m_0}$ on the diagonal and its determinant is $1$ if and only if $m_0$ is odd. 
	Therefore, if $\lambda$ is even, then $M^{-1} \, \text{mod}\, (E, \, \phi)$ exists if and only if $m_0$ is odd.
\end{proof}
\noindent \\
A basis of the lattice $\mathcal{L}$ is:
\begin{center}
	$\mathcal{M}_1 =
	\begin{pmatrix} 
	p & 0 & 0 & \dots & 0 & 0 \\
	t_1 & 1 & 0 & \dots & 0 & 0 \\
	t_2 & 0 & 1 & \dots & 0 & 0 \\
	\vdots & & & \ddots & & \vdots \\
	
	t_{n-2} & 0 & 0 & \dots & 1 & 0 \\
	t_{n-1} & 0 & 0 & \dots & 0 & 1 \\
	\end{pmatrix}$ 
\end{center} 
where  $t_i = (-\gamma)^i \, \text{mod}\, (p)$. 
\\
\textit{Note:} Here, a polynomial is represented by a row vector corresponding to its coefficients, i.e.  if $A(X) = a_0 + a_1X + \dots + a_{n-1}X^{n-1}$, then the corresponding vector is: $(a_0, \,a_1, \,\dots, \,a_ {n-1})$.
\begin{proposition}
	Let $\mathcal{G}$ be a reduced basis of the lattice $\mathcal{L}$ obtained from the basis $\mathcal{M}_1$. Then, at least one (row) vector $\mathcal{G}_i$ of  $\mathcal{G}$ is such that $\mathcal{G}_{i,0}$ is odd, (i.e  $\mathcal{G}_i$ is a suitable candidate for $M$, according to proposition \ref{prop:even-m-ex}). 
\end{proposition}
\begin{proof}
	$\mathcal{G}$ is a basis of $\mathcal{L}$, so the vector $P = (p, \,0, \,\dots, \,0)$ (i.e  the first line of $\mathcal{M}_1$) is a linear combination (over $\mathbb{Z}$) of elements of $\mathcal{G}$. Now, suppose that the first component of all elements of $\mathcal{G}$ is even. This means that every linear combination  of elements from $\mathcal{G}$  will output a vector whose first component is even. As $p$ is odd, this is in contradiction with the fact that $P \in \mathcal{L}$ and $\mathcal{G}$ is a basis of $\mathcal{L}$. Thus, at least one element $\mathcal{G}_i$ of $\mathcal{G}$ must be such that $\mathcal{G}_{i,0}$ is odd.
\end{proof}

\subsubsection*{\textit{{Case 2: $\lambda$ is odd}}}
\begin{proposition}
	Let $R_n = \mathbb{F}_2[X]/(X^n - 1)$ be the algebra of all polynomials modulo $(X^n - 1)$ over $\mathbb{F}_2$.
	Let $C_n$ be the ring of $n \times n$ binary circulant matrices. \\ 
	Let $\Phi$ be the application defined as follows:
	\begin{center}
		\begin{small}
			$\begin{array}{ccccc}
			\Phi & : & R_n & \to & C_n \\
			& & a_0 + \dots + a_{n-1}X^{n-1} & \mapsto & \begin{pmatrix} 
			a_0 & a_1 & \dots & a_{n-2} & a_{n-1} \\
			a_{n-1} & a_0 & \dots & a_{n-3} & a_{n-2} \\
			\vdots & & \ddots & & \vdots \\
			a_2 & a_3 & \dots & a_0 & a_1 \\
			a_1 & a_2 & \dots & a_{n-1} & a_0 \\
			\end{pmatrix} \\
			\end{array}$ 
		\end{small}
	\end{center}
	Then, $R_n$ is isomorphic to $C_n$ and $\Phi(R_n) = C_n$. 
\end{proposition}
\begin{proof}
	see section 3.4 of \cite{baldi2014qc}.
\end{proof}
\begin{corollary}
	\label{rmk:invm}
	Let $A \in R_n$ be a polynomial, the determinant of $\Phi(A)$ is odd if and only if $gcd(A, X^n - 1) = 1$.
\end{corollary}
\begin{proof}
	$A$ is invertible if and only if $gcd(A, X^n - 1) = 1$. An element in $C_n$ is invertible if and only if its determinant is odd. As $R_n$ is isomorphic to $C_n$ through $\Phi$, this concludes the proof.
\end{proof}
When $\lambda$ is odd, the determinant of $\mathcal{S}_{M,E}$ has the same parity than the determinant of the following matrix  (using  property \ref{rmk:mat_parity}), where $\overline{m_i} = m_i\pmod 2$: 
\begin{center}
	\begin{small}
		$
		\begin{pmatrix} 
		1 & 0 & \dots & 0 & \overline{m_{n-1}} & 0 & \dots & 0 & 0\\
		0 & 1 & \dots & 0 & \overline{m_{n-2}} & \overline{m_{n-1}} & \dots & 0 & 0 \\
		\vdots & & \ddots & & \vdots & & & & \vdots \\
		0 & 0 & \dots & 1 & \overline{m_1} & \overline{m_2} & \dots & \overline{m_{n-1}} & 0 \\
		0 & 0 & \dots & 0 & \overline{m_0} & \overline{m_1} & \dots & \overline{m_{n-2}} & \overline{m_{n-1}} \\
		-1 & 0 & \dots & 0 & 0 & \overline{m_0} & \dots & \overline{m_{n-3}} & \overline{m_{n-2}} \\
		\vdots & & & & \vdots & & \ddots & & \vdots \\
		0 & 0 & \dots & 0 & 0 & 0 & \dots & \overline{m_0} & \overline{m_1} \\
		0 & 0 & \dots & -1 & 0 & 0 & \dots & 0 & \overline{m_0} \\
		\end{pmatrix}$ 
	\end{small}
\end{center}
Using the fact that the addition of one row to another row doesn't change the value of the determinant, one deduces that the determinant of the preceding matrix is the same than the determinant of the matrix: 
\begin{center}
	\begin{small}
		$
		\begin{pmatrix} 
		1 & 0 & \dots & 0 & \overline{m_{n-1}} & 0 & \dots & 0 & 0\\
		0 & 1 & \dots & 0 & \overline{m_{n-2}} & \overline{m_{n-1}} & \dots & 0 & 0 \\
		\vdots & & \ddots & & \vdots & & & & \vdots \\
		0 & 0 & \dots & 1 & \overline{m_1} & \overline{m_2} & \dots & \overline{m_{n-1}} & 0 \\
		0 & 0 & \dots & 0 & \overline{m_0} & \overline{m_1} & \dots & \overline{m_{n-2}} & \overline{m_{n-1}} \\
		0 & 0 & \dots & 0 & \overline{m_{n-1}} & \overline{m_0} & \dots & \overline{m_{n-3}} & \overline{m_{n-2}} \\
		\vdots & & & & \vdots & & \ddots & & \vdots \\
		0 & 0 & \dots & 0 & \overline{m_2} & \overline{m_3} & \dots & \overline{m_0} & \overline{m_1} \\
		0 & 0 & \dots & 0 & \overline{m_1} & \overline{m_2} & \dots & \overline{m_{n-1}} & \overline{m_0} \\
		\end{pmatrix}$ 
	\end{small}
\end{center}
This matrix has the same determinant as the matrix:
\begin{center}
	\begin{small}
		$\mathcal{H} =
		\begin{pmatrix} 
		\overline{m_0} & \overline{m_1} & \dots & \overline{m_{n-2}} & \overline{m_{n-1}} \\
		\overline{m_{n-1}} & \overline{m_0} & \dots & \overline{m_{n-3}} & \overline{m_{n-2}} \\
		\vdots & & \ddots & & \vdots \\
		\overline{m_2} & \overline{m_3} & \dots & \overline{m_0} & \overline{m_1} \\
		\overline{m_1} & \overline{m_2} & \dots & \overline{m_{n-1}} & \overline{m_0} \\
		\end{pmatrix}$ 
	\end{small}
\end{center}
Therefore, if $\lambda$ is odd, the determinant of $\mathcal{S}_{M,E}$ and $\mathcal{H}$ have the same parity.
\begin{definition}
	\label{def:over-pol}
	Let $P \in \mathbb{Z}[X]$ be a polynomial such that $P(X) = p_0 + p_1X + \dots + p_{n-1}X^{n-1}$.
	We denote by $\overline P$ the polynomial in $R_n$ such that: $\overline P(X) = p'_0 + p'_1X + \dots + p'_{n-1}X^{n-1}$ where $p'_i = p_i \, \text{\rm mod} \,(2)$.
\end{definition}
\begin{proposition}
	\label{prop:odd-m-ex}
	Let $E(X) = X^n - \lambda$ such that $\lambda$ is odd. Let $M = m_0 + m_1X + \dots + m_{n-1}X^{n-1}$ be a polynomial such that $M \in \mathcal{L}$. Then $M^{-1} \, \text{mod}\, (E, \, \phi)$ exists if and only if $gcd(\overline M, X^n - 1) = 1$.
\end{proposition}
\begin{proof}
	The circulant matrix $\mathcal{H}$, defined above, is such that $\mathcal{H} = \Phi(\overline M)$. As $\overline M \in R_n$, the determinant of $\mathcal{H}$ is odd iff $gcd(\overline M, X^n - 1) = 1$, using corollary \ref{rmk:invm}.
	If $\lambda$ is odd, we showed above that the determinant of $\mathcal{S}_{M,E}$ and $\mathcal{H}$ have the same parity. So, $Res(E,M)$ is odd iff $gcd(\overline M, X^n - 1) = 1$. That is, $M^{-1} \, \text{mod}\, (E, \, \phi)$ exists if and only if $gcd(\overline M, X^n - 1) = 1$.
\end{proof}
\noindent\\
With a well chosen basis of $\mathcal{L}$, we are now going to show that there always exists a small vector $M$ such that $gcd(\overline M, X^n - 1) = 1$. \\
Let $\mathcal{M}_2$ be a basis of the lattice $\mathcal{L}$ such that:
\begin{center}
	$\mathcal{M}_2 =
	\begin{pmatrix} 
	p & 0 & 0 & \dots & 0 & 0 \\
	s_1 & 1 & 0 & \dots & 0 & 0 \\
	s_2 & 0 & 1 & \dots & 0 & 0 \\
	\vdots & & & \ddots & & \vdots \\
	s_{n-2} & 0 & 0 & \dots & 1 & 0 \\
	s_{n-1} & 0 & 0 & \dots & 0 & 1 
	\end{pmatrix}$ 
\end{center} 
where: $s_i = t_i + p.k_i$ with $t_i = (-\gamma)^i \, \pmod p$ and $k_i = t_i \pmod 2$.
Notice that all $s_i$ are even (as $p$ is odd).
\begin{proposition}
	Let $\mathcal{G}=\{\mathcal{G}_1,\ldots,\mathcal{G}_n\}$ be a reduced basis of the lattice $\mathcal{L}$ obtained from the basis $\mathcal{M}_2$. Then, there exists a linear combination 
	$(\beta_1,\ldots,\beta_n)$ with $\beta_i\in\{0,1\}$ such that $M=\sum_{i=1}^n\beta_i\mathcal{G}_i$ satisfies gcd$(\overline M, X^n - 1) = 1$. Hence  $M^{-1} \, \text{mod}\, (E, \, \phi)$ exists.
\end{proposition}
\begin{proof}
	First, we have that: 
	\begin{center}
		$\overline{\mathcal{M}_2} =
		\begin{pmatrix} 
		1 & 0 & 0 & \dots & 0 & 0 \\
		0 & 1 & 0 & \dots & 0 & 0 \\
		0 & 0 & 1 & \dots & 0 & 0 \\
		\vdots & & & \ddots & & \vdots \\
		0 & 0 & 0 & \dots & 1 & 0 \\
		0 & 0 & 0 & \dots & 0 & 1
		\end{pmatrix}$ 
	\end{center} 
	where $\overline{\mathcal{M}_2}_{ij} = {\mathcal{M}_2}_{ij}\pmod 2$. Each line $i$ of $\overline{\mathcal{M}_2}$ corresponds to the polynomial $X^i \in R_n$, for $0 \leqslant i < n$. This means that $\overline{\mathcal{M}_2}$ is a basis of $R_n$. \\
	Let $U \in R_n$ be a polynomial such that $gcd(U, X^n -1) = 1$.
	As $\overline{\mathcal{M}_2}$ is a basis of $R_n$, it exists $T = (t_1, \dots, t_{n}) \in \mathbb{F}_2^n$ such that $U = T.\overline{\mathcal{M}_2}$. As $T.\overline{\mathcal{M}_2} = \overline{T.\mathcal{M}_2}$, we obtain that $U = \overline{T.\mathcal{M}_2}$. \\
	We have $T.\mathcal{M}_2 \in \mathcal{L}$, so there exists $V = (v_1, \dots, v_{n}) \in \mathbb{Z}^n$ such that $V.\mathcal{G} = T.\mathcal{M}_2$, as $\mathcal{G}$ is a basis of $\mathcal{L}$. Thus, $U = \overline{V.\mathcal{G}}$. \\
	Let $\beta = (\beta_1, \dots, \beta_n) \in \mathbb{F}_2^n$ such that $\beta_i = v_i \pmod 2$, then we have $U = \overline{\beta.\mathcal{G}}$. \\
	Let $M \in \mathcal{L}$ be a polynomial such that $M = \sum_{i=1}^n \beta_i\mathcal{G}_i$, then $\overline{M} = U$, hence gcd$(\overline{M}, X^n -1) = 1$. Using proposition \ref{prop:odd-m-ex}, this proves that  $M^{-1} \, \text{mod}\, (E, \, \phi)$ exists.
\end{proof}
\noindent 
\begin{remark}
	From the preceeding proposition, at most $2^n$ linear combinations of elements of $\mathcal{G}$ must be computed in order to find a suitable polynomial $M$.
	For cryptographic sizes, $n$ is small enough so that one can check all these combinations  (see examples in Appendix \ref{annex:amns}, for some possible values of $n$).
	Moreover, from proposition \ref{prop:odd-m-ex}, if a combination leads to a polynomial  $M$ with an even number of odd coefficients, then $gcd(\overline M, X^n - 1)$ is a multiple of $X-1$ (as $1$ is a root of $\overline M$ in this case). Thus, $M^{-1} \, \text{mod}\, (E, \, \phi)$ doesn't exist. Hence, such polynomials $M$ can be discarded.
\end{remark}
\begin{remark}
	Let $\theta= \max\limits_{0 \leqslant i < n} \|\mathcal{G}_i\|_{\infty}$. For any of the $2^n$ linear  combinations, the corresponding polynomial $M$ verifies $\|M\|_{\infty} \leqslant n\theta$. Thus, if elements of $\mathcal{G}$ are small, then $\|M\|_{\infty}$ will also be small, as $n$ is small (and negligible compared to $\theta$). 
\end{remark} 
\subsection{Computation of representatives of powers of $\rho$} \label{subsec:computRho}
To compute the representatives of $\mathit{\rho^i}$, for $1 \leqslant i < n$, needed for Algorithm \ref{alg:conv_to_amns}, 
we first use Algorithm \ref{alg:conv_to_amns2}, with $\rho$ as input, to find a representative $R$ of $\rho.\phi$ in the AMNS. Then, using RedCoeff with $R$ as input, we obtain a representative of $\rho$. For any $i \geqslant 2$, we compute a representative of $\rho^i$ by multiplying a representative of $\rho^{i-1}$ by $R$, using Algorithm \ref{alg:amns_mult}.

\section{Implementation results}

\subsection{Theoretical performances and memory consumption}
\label{subsec:th-perf-mem}
The performances and the memory consumption of an AMNS depend mainly on the target architecture and the value of $n$ which corresponds to the number of terms of the underlying polynomials of the AMNS. The time and space complexity of the AMNS operations mainly depend on $n$.

Let's consider that we have a $k$-bits processor architecture.
In order to obtain an efficient AMNS, a good idea is to take $\phi = 2^k$ because in that case, division and modular reduction by $\phi$ can be done with very simple mask and shift operations. This choice makes the reduction modulo $\phi$ at line 1 and the division at line 3 of algorithm \ref{alg:redcoeff} very cheap. This way, $nk$ bits will be used to store the AMNS representative of an integer. The product of two elements involves $\mathcal{O}(n^2)$ arithmetic operations over $k$-bits integers and, using the trick mentionned in subsection \ref{subsec:addition}, adding two elements needs exactly $n$ additions of $k$-bits integers.

\subsection{Practical performances and memory consumption}
\label{subsec:pr-perf-mem}
\subsubsection{About implementation: a C code generator.}
Despite all the parameters needed for arithmetic operations in an AMNS, its software implementation is quite easy.
To prove that, we wrote with the SageMath library \cite{sagemath} a code that generates a C code for any AMNS, given its complete set of parameters. \\
Our implementations of AMNS generation, the C code generator and the AMNS we used for our numerical experimentation below are avalaible on {GitHub:}\\\verb+https://github.com/eacElliptic/AMNS+.

\subsubsection{Computer features.}
For our tests we used a Dell Precision Tower 3620 on Ubuntu gnome 16.04-64 bits with  an Intel Core i7-6700 processor and 32GB RAM. We compiled our tests with gcc 5.4 using O3 compiling option and we compared our results to GNU MP 6.1.1 and OpenSSL 1.0.2g implementations. These libraries have also been compiled with gcc with O3 option.

\subsection{Numbers of AMNS for a given prime}
An interesting but complex question is how many AMNS can be generated given a prime number for a target architecture. This question is difficult to answer because of the large range of parameters that define an AMNS and also because it is linked to the existence of a nth-root $\gamma$ of a given $\lambda$ in $\mathbb{Z}/p\mathbb{Z}$.\\
Here, we give an answer while focusing on the efficiency of arithmetic operations. This will lead us to add some constraints which will of course reduce the number of AMNS.

Let us first assume that we have a $k$-bits architecture. Once $p$ is known, we have to choose the parameters $n$, $\phi$ and $\lambda$.
For any AMNS, we have: $\phi \geqslant 2|\lambda|n\rho$ \, and \, $(2\rho)^n \geqslant p$.
So, $\log_{2}|\lambda| \leqslant \log_{2}\phi - \log_{2}n - (\log_{2}p)/n$.
As mentionned in subsection \ref{subsec:th-perf-mem}, 
for performance reasons, we choose $\phi = 2^k$. Hence, $n > \frac{\log_{2}p}{k}$ (as $\phi > \rho$). Remember that $n$ is the number of $k$-bits words used to represent an integer. Hence we want $n$ as small as possible to minimize the computation. 
Let's assume that we choose $n$ such that $$\frac{\log_{2}p}{k} + 1 + c \geqslant n > \frac{\log_{2}p}{k},$$ 
i.e we allow at max $c$ more coefficients than the optimal value $\lfloor\frac{\log_{2}p}{k}\rfloor + 1$. According to subsection \ref{subsec:th-perf-mem}, the smaller $c$ is, the better performances and memory consumption will be. Therefore, $\log_{2}n > \log_{2}\log_{2}p - \log_{2}k$ and $\frac{\log_{2}p}{n} \geqslant \frac{k\log_{2}p}{\log_{2}p + kc + k}$. So, we finally have: 
\begin{center}
$\log_{2}|\lambda| < k + \log_{2}k - \log_{2}\log_{2}p - \frac{k\log_{2}p}{\log_{2}p + k.c + k}$
\end{center}
Let $\Omega = k + \log_{2}k - \log_{2}\log_{2}p - \frac{k\log_{2}p}{\log_{2}p + kc + k}$. To efficiently compute the reduction modulo $X^n-\lambda$ (external reduction process), we choose $\lambda$ such that $\lambda = \pm 2^i \pm 2^j$, where $0 \leqslant i,j < \Omega$ and $i \neq j$.
Consequently, the number of values for $\lambda$ is $4  \binom{\Omega}{2}$. 

Now, the main difficulty is to know how many nth-root modulo $p$ can be computed with this set of values for $\lambda$. Notice that such a root may not exist for some values of $\lambda$. To give an answer, we distinguish two case according to $\text{gcd}(n, p-1)$.
\subsubsection{Case 1 : gcd$(n,p-1) = 1$.} Using corollary \ref{coro:lmbd}, we obtain that any value $\lambda \in (\mathbb{Z}/p\mathbb{Z})^{*}$, except $1$, gives a useful nth-root modulo $p$. So, in this case, one can generate at least $4 \binom{\Omega}{2} - 1$ AMNS. 
\subsubsection{Case 2 : gcd$(n,p-1) > 1$.} This case is complex because finding whether or not a value $\lambda \neq 1$ has a nth-root modulo $p$ could lead to solve a hard instance of the discrete logarithm problem (see proposition \ref{prop-gmm}). But, using corollary \ref{coro:lmbd1}, we obtain that one can generate at least gcd$(n, p-1) - 1$ AMNS (taking $\lambda = 1$).
\begin{remark}
In both cases above, we gave the minimum numbers of AMNS that can be generated. 
Indeed, once $\lambda$ and $\gamma = \lambda^{1/n} \pmod p$ are fixed, the next parameter to compute is the polynomial $M$ using lattice reduction. From subsection \ref{subsec:genparam}, at least one polynomial in a reduced basis of the lattice satisfies the required constraints on $M$. In fact, our practical experimentations have shown that there are, most of the time, more than one suitable candidate for $M$. Moreover, some linear combinations of the polynomials of the reduced basis give suitable candidates for $M$.
For a tuple $(p, n, \lambda, \gamma)$, different suitable candidates for $M$ lead to different AMNS. So, one can generate much more AMNS than the minimum numbers we gave using some linear combinations of the polynomials of the reduced basis.
\end{remark} 
\begin{example}
\label{expl:nb-amns}
We generated a set of AMNS for some primes of size 192, 224, 256, 384 and 521 bits. These sizes correspond to the NIST recommended key sizes for elliptic curve cryptography. For this experiment, we took $k=64$ and $c=2$ (so, we allow at max $2$ more coefficients than the optimal value $\lfloor\frac{\log_{2}p}{64}\rfloor + 1$, which is quite restrictive but good for performances). Remember that $\Omega = k + \log_{2}k - \log_{2}\log_{2}p - \frac{k\log_{2}p}{\log_{2}p + kc + k}$. \\
As already said, the main difficulty in the generation process is the computation of $\gamma$ a nth-root modulo $p$ of $\lambda$. For our test, we used the library SageMath which uses the algorithm proposed in \cite{nthroot} for this operation. As it can take a lot of time, we fixed in our code a timeout of thirty minutes. So, some values of $\lambda$ that have nth-roots might have been discarded. Finally, to extend the numbers of AMNS according to the remark above, we only checked all binary combinations of lattices elements; a larger set should lead to more AMNS. \\ 
With these parameters and constraints, the table \ref{tab:nb-amns} gives the number of AMNS we found for each prime. We call these primes p192, p224, p256, p384 and p521 according to their bit sizes (see Appendix \ref{annex:prime-list} for their values). The recommended NIST prime $p = 2^{521} - 1$ is denoted by nist\_p521.

\begin{table}[h]
\renewcommand{\arraystretch}{1.4}
\begin{center}
\small
\begin{tabular}{ll}
\begin{tabular}{|l|c|c|c|}
\hline
Prime number & p192 & p224 & p256 \\
\hline
Number of AMNS & $10418$ & $5118$ & $11877$ \\
\hline
\end{tabular}
\\\\
\begin{tabular}{|l|c|c|c|}
\hline
Prime number & p384 & p521 & nist\_p521 \\
\hline
Number of AMNS & $14787$ & $19871$ & $85592$ \\
\hline
\end{tabular}
\end{tabular}
\end{center}
\caption{\label{tab:nb-amns}A lower bound on the number of distinct AMNS for some prime integers, using a timer of 30 minutes for nth-root computation.}
\end{table}
\end{example}

\subsubsection{Numerical experimentation.}
For each NIST recommended key size for elliptic curve cryptography (i.e 192, 224, 256, 384 and 521 bits), we have generated a set of AMNS for many primes of that size, see Appendix \ref{annex:amns} for some exemples of AMNS.
As in example \ref{expl:nb-amns}, we took $k = 64$ and $\phi = 2^k$. So, for software implementation, the optimal value for $n$ is $n_{opt} = \lfloor\frac{\log_{2}p}{k}\rfloor + 1$. We also took $c = 2$, which means that for each prime, we have generated AMNS with $n$ equals to $n_{opt}$, $n_{opt} + 1$ and $n_{opt} + 2$. Then, with each AMNS, we have computed $2^{25}$ modular multiplications using the AMNS representatives of random elements. With the same inputs, we have also computed the $2^{25}$ modular multiplications with the well known librairies GNU MP \cite{gnu_mp} and OpenSSL \cite{openssl}. For OpenSSL, we used the default modular multiplication procedure and Montgomery modular multiplication already implemented. \\

Table \ref{tab:perf} gives the ratio between the performances obtained for the AMNS and the performances obtained for GNU MP and OpenSSL. We computed these ratios for $n$ equals to $n_{opt}$, $n_{opt} + 1$ and $n_{opt} + 2$. 

Table \ref{tab:mem} gives the memory consumption to store an integer modulo $p$ using the {\small AMNS}, {\small GNU MP} and {\small OpenSSL}, where $n = n_{opt}$ for the AMNS. More precisely, we give the number of $64$-bits integers used to represent the elements of $\mathbb{Z}/p\mathbb{Z}$.
For GNU MP and OpenSSL, we give in Appendix \ref{annex:memory} the source code corresponding to the structure used for storing an integer with these libraries.  

\begin{table}[h]
\footnotesize
\begin{tabular}{ll}
\begin{tabular}{|l|c|c|c|c|c|c|c|c|c|}
\hline
$p$ size & \multicolumn{3}{c|}{192} & \multicolumn{3}{c|}{224} \\
\hline
\textbf{n} & \textbf{4} & \textbf{5} & \textbf{6} & \textbf{4} & \textbf{5} & \textbf{6} \\
\hline
ratio 1 & \textbf{0.86} & 1.41 & 2.04 & \textbf{0.57} & \textbf{0.98} & 1.41 \\
\hline
ratio 2 & 0.10 & 0.17 & 0.24 & 0.08 & 0.14 & 0.19 \\
\hline
ratio 3 & 0.21 & 0.34 & 0.49 & 0.16 & 0.27 & 0.39 \\
\hline    
\end{tabular}
\\\\
\begin{tabular}{|l|c|c|c|c|c|c|}
\hline
$p$ size & \multicolumn{3}{c|}{256} &\multicolumn{3}{c|}{384}\\
\hline
\textbf{n} & \textbf{5} & \textbf{6} & \textbf{7}  & \textbf{7} & \textbf{8} & \textbf{9} \\
\hline
ratio 1 &\textbf{0.98} & 1.42 & 1.84 & \textbf{0.98} & 1.34 & 1.67  \\
\hline
ratio 2 & 0.14 & 0.20 & 0.26 & 0.19 & 0.25 & 0.31 \\
\hline
ratio 3 & 0.30 & 0.43 & 0.55 & 0.43 & 0.58 & 0.73 \\
\hline    
\end{tabular}
\\\\
\begin{tabular}{|l|c|c|c|c|c|c|}
\hline
$p$ size & \multicolumn{3}{c|}{521} \\
\hline
\textbf{n} & \textbf{10} & \textbf{11} & \textbf{12} \\
\hline
ratio 1 & \textbf{0.95} & 1.18 & 1.36 \\
\hline
ratio 2 & 0.25 & 0.29 & 0.34 \\
\hline
ratio 3 & 0.56 & 0.69 & 0.80 \\
\hline    
\end{tabular}
\end{tabular}
\vspace{2mm}

\footnotesize
 ratio 1: AMNS/OpenSSL Montgomery modular mult. \\
 ratio 2: AMNS/OpenSSL default modular mult.\\
 ratio 3: AMNS/GNU MP mult. + modular reduction
\caption{\label{tab:perf}Relative performances of AMNS vs GNU MP and OpenSSL modular multiplications, with $n$ equals to $n_{opt}$, $n_{opt} + 1$ and $n_{opt} + 2$ for the AMNS.}
\end{table}

As it can be seen in table \ref{tab:perf}, the AMNS performs modular multiplication much more efficiently than the librairy GNU MP and the default method in OpenSSL for all values of $n$ in the table. We can also observe that AMNS modular multiplication is slightly faster than the Montgomery method in OpenSSL when the value of $n$ is optimal.\\
\noindent\\
Table \ref{tab:perf} gives the mean ratio for each pair \textit{(size of $p$, n)}. In table \ref{tab:perf-best}, we give the best ratios obtained when $n = n_{opt}$ for AMNS. These ratios are obtained when $\lambda$ is very small with good shape. For example, $\pm 1$ or $\pm 2$. \\

For size of $p$ and $n$ equals to 224 and $4$ respectively, it can be observed that the obtained ratio overperforms the others ratios. In this case, to represent an integer of $\mathbb{Z}/p\mathbb{Z}$, AMNS, GNU MP and OpenSSL use the same number of $k$-bits blocks. As elements are polynomials in AMNS, there is no carry to manage. Thus, arithmetic operations are a lot faster. In the other cases, AMNS use at least one more block than GNU MP and OpenSSL.

\begin{remark}
When the size of $p$ is 521 bits, the optimal value for $n$ is 9. With our constraints on $k$, $\phi$, $\lambda$ and the timer we used to compute a nth-root modulo $p$ of $\lambda$, we did not find an AMNS with $n = 9$. So, for this size, ratios where computed with $n$ equals to $n_{opt} + 1$, $n_{opt} + 2$ and $n_{opt} + 3$.
\end{remark}

\begin{table}[h]
\footnotesize
\begin{tabular}{ll}
\begin{tabular}{|l|c|c|c|}
\hline
($p$ size, $n$) & (192, 4) & (224, 4) & (256, 5)\\
\hline
ratio 1 & \textbf{0.77} & \textbf{0.56} & \textbf{0.91} \\
\hline
ratio 2 & 0.09 & 0.08 & 0.13 \\
\hline
ratio 3 & 0.19 & 0.16 & 0.28 \\
\hline    
\end{tabular}
\\\\
\begin{tabular}{|l|c|c|}
\hline
($p$ size, $n$) & (384, 7) & (521, 10) \\
\hline
ratio 1 & \textbf{0.92}  & \textbf{0.91} \\
\hline
ratio 2 & 0.18 & 0.24 \\
\hline
ratio 3 & 0.40 & 0.54 \\
\hline    
\end{tabular}
\end{tabular}
\vspace{2mm}

\footnotesize
 ratio 1: AMNS/OpenSSL Montgomery modular mult. \\
 ratio 2: AMNS/OpenSSL default modular mult.\\
 ratio 3: AMNS/GNU MP mult. + modular reduction
\caption{\label{tab:perf-best}Relative performances of AMNS vs GNU MP and OpenSSL modular multiplications, with $n$ equals to $n_{opt}$ for the AMNS (best ratios).}
\end{table}

Moreover, in our implementations, we did not take advantage of the high parallelisation capability of AMNS, which should make arithmetic operations much faster. This parallelisation capability comes from the polynomial structure of the elements in AMNS. \\
In fact, using parallelisation, one could divide by $n$ the execution time of 
line 1 of algorithm \ref{alg:amns_mult} as the coefficients of the result $V$ can be computed independently in the same time. Likewise, the same thing can be done with all the lines of the RedCoeff method (algorithm \ref{alg:redcoeff}), which is called in algorithm \ref{alg:amns_mult}. 

An AMNS addition should always be faster because it is simply a polynomial addition without carries to manage, unlike classic binary representations. It should even be better with parallelisation.\\

\begin{remark}
Although, AMNS is a lot faster than GNU MP and OpenSSL default method for modular multiplication, the most relevant ratio in table \ref{tab:perf} is ratio1 because AMNS requires roughly the same amount of data to precompute than the Montgomery modular multiplication.
\end{remark}

\begin{table}[h]
\renewcommand{\arraystretch}{1.4}
\begin{center}
\small
\begin{tabular}{|l|c|c|c|c|c|}
\hline
Size in bits of $p$ & 192 & 224 & 256 & 384 & 521 \\
\hline
AMNS & 4 & 4 & 5 & 7 & 10 \\
\hline
GNU MP (mpz\_t) & 4 & 5 & 5 & 7 & 10 \\
\hline
OpenSSL (BIGNUM) & 5 & 6 & 6 & 8 & 11 \\
\hline
\end{tabular}
\end{center}
\caption{\label{tab:mem}Number of 64 bits words used to store elements of $\mathbb{Z}/p\mathbb{Z}$, where $n$ equals to $n_{opt}$ for the AMNS.}
\end{table}

In table \ref{tab:mem}, it can be observed that the memory consumptions of AMNS and GNU MP are most of the time the same while OpenSSL consumes more memory.\\Notice that for AMNS and OpenSSL Montgomery, some data must be precomputed, but the memory consumptions of these data are negligible compared to the overall usage of memory when performing multiple arithmetic operations.

\begin{remark}
In special cases where the polynomials $M$ and $M'$ of an AMNS have many coefficients equal to zero, the RedCoeff (algorithm \ref{alg:redcoeff}) method becomes very fast. This leads to a very efficient modular multiplication process.
As an example, for the recommended NIST prime $p = 2^{521} - 1$, we found such an AMNS (see Appendix \ref{annex:nist-521} for its parameters). 
With this AMNS, using the same inputs for $2^{25}$ iterations of modular multiplications, we obtained interesting timings that are collected in table \ref{tab:perf2}.
\begin{table}[h]
\begin{center}
\begin{tabular}{ll}
\begin{tabular}{|c|c|c|c|c|c|}
\hline
\multicolumn{3}{|c|}{$p = 2^{521} - 1$, $n = 10$} \\
\hline
ratio1 & ratio2 & ratio3 \\
\hline
0.42 & 0.09 & 0.27 \\
\hline    
\end{tabular}
\end{tabular}
\end{center}
\vspace{1mm}
\small
 ratio 1: AMNS/OpenSSL Montgomery modular mult. \\
 ratio 2: AMNS/OpenSSL default modular mult.\\
 ratio 3: AMNS/GNU MP mult. + modular reduction
\caption{\label{tab:perf2}Example of relative performances of AMNS vs GNU MP and OpenSSL modular multiplications when reduction polynomials are sparse}
\end{table}

This efficiency for this specific modulus is due to the fact that the two polynomial multiplications by $M$ and $M'$ in the RedCoeff method are very cheap as these polynomials have most of their coefficients equal to zero.\\
So, independently of the value of $n$, the sparser are $M$ and $M'$, the better are the performances.
\end{remark}

\subsection{About side channel attacks}
\label{subsec:sca}
AMNS have very interesting properties regarding side channel attacks. 
\subsubsection{Regular algorithms.}
All the described algorithms (conversions, addition and multiplication) contain no conditional branching which is one of the basic weaknesses used in some side channel attacks, like the simple power analysis (SPA). 
\subsubsection{Many AMNS for  a given prime.}
Given a tuple $(p, n, \lambda)$, if a nth-root (modulo $p$) $\gamma$ of $\lambda$ exists, we know from proposition \ref{prop-gmm} that the total number of such roots is gcd$(n, p-1)$ and each of these roots allows to build at least one AMNS for the same prime $p$.
Moreover, corollaries \ref{coro:lmbd} and \ref{coro:lmbd1} show that for any pair $(p, n)$, it is always possible to choose $\lambda \in \mathbb{Z}$, and easily find $\gamma$, such that $\gamma^n \pmod p \equiv \lambda$. This means that it is always possible to generate many AMNS, given a prime $p$.
In fact, with our implementation, we were able to generate thousands of AMNS for many modulus using restrictive conditions (see table \ref{tab:nb-amns} in example \ref{expl:nb-amns}). This property is very interesting because, most of the time, side channel attacks use patterns (and hypothesis) to find secret data. With many AMNS for a given modulus, it becomes difficult to build (or find) such patterns as it is expected that any AMNS will have its specific behaviour for arithmetic operations. 
\subsubsection{Unpredictable shape of a representative.}
Finally, the difference between the representatives of an element from one AMNS to another will make side channel attacks much more complex.
As an example, we generated three AMNS for the prime $p = 2^{255} + 95$, and considered the representatives of the element \\ $t = 4D9B499C5B883B0F11752FBEED0684B6972\break F588DB67810835002A07C2F2AC804 \in \mathbb{Z}/p\mathbb{Z}$ in these AMNS; see Appendix \ref{annex:amns-ex-sm-p} for their parameters.\\\\
In AMNS 1, a representative of $t$ is: $X^3$. \\\\
In AMNS 2, a representative of $t$ is: \\
$-0x39CDC4224C412.X^4 - 0x3F60F0AC55927.X^3 + 0x7BD09DBD01E4.X^2 - 0x4B844B52F420E.X -\break 0x2DE4B18019BCF$.\\\\
In AMNS 3, a representative of $t$ is: \\
$-0x152DDD219CC.X^5 + 0x3080218E9FD.X^4 - \break 0x225BF4D6DE9.X^3 +0x672A6C1F62E.X^2 - \break 0x3E6242D6F01.X + 0x7DEE4F8F11E$. \\\\
As it can be observed, these representatives of $t$ vary a lot. This behaviour should make side channel attacks much more complex as it will introduce some kind of randomness in power consumptions for example.
\subsubsection{A simple DPA countermeasure for ECC.}
In \cite{joye-timen}, the authors show how to randomize the base point $P$
to thwart DPA attacks. The main idea is to change the $(x,y)$ coordinates of $P$ by $(u^{-2}x,u^{-3}y)$ for a random $u\in(\mathbb{Z}/p\mathbb{Z})^\times$. Such a countermeasure can easily be implemented in the AMNS conversion procedure.
The conversion algorithms (from binary representation to AMNS and vice-versa) are modified to take an extra argument used to change the representative of $x$ and $y$ (see algorithms \ref{alg:dpa_conv_to_amns} and \ref{alg:dpa_conv_from_amns}). Hence before the computation of $kP$, the procedures DPA\_Conv\_2\_AMNS($x$,$u^{-2}$) and DPA\_Conv\_2\_AMNS($y$,$u^{-3}$) are called. Once the computation done a call to DPA\_Conv\_2\_BIN($x$,$u^{2}$) and DPA\_Conv\_2\_BIN($y$,$u^{3}$) allows to find back the coordinates of $kP$.
\begin{algorithm}
  \caption{DPA\_Conv\_2\_AMNS(a,$\beta$)}
  \label{alg:dpa_conv_to_amns}
  \begin{algorithmic}[1]
    \Require $a \in \mathbb{Z}/p\mathbb{Z}$ \, $\mathcal{B} = (p, \, n, \, \gamma, \, \rho, E)$, and $\beta\in(\mathbb{Z}/p\mathbb{Z})^{\times}$
    \vspace{1mm}
    \Ensure $A \equiv (a\beta.\phi)_\mathcal{B}$
    \vspace{1mm}
    \State $b = (a\beta.\phi^2) \, \text{mod}\, p$
    \vspace{1mm}
    \State $b = (b_{n-1}, ..., b_0)_\rho$
	\vspace{1mm}    
    \State $U \leftarrow \sum\limits_{i=0}^{n-1} b_i.P_i(X)$
    \vspace{1mm}
    \State $A \leftarrow \text{RedCoeff}(U)$
    \vspace{1mm}
    \State return $A$
  \end{algorithmic}
\end{algorithm}
\begin{algorithm}
  \caption{DPA\_Conv\_2\_BIN(A,$\beta$)}
  \label{alg:dpa_conv_from_amns}
  \begin{algorithmic}[1]
    \Require $A \in \mathcal{B}$, $\mathcal{B} = (p, \, n, \, \gamma, \, \rho, E)$ and  $\beta\in(\mathbb{Z}/p\mathbb{Z})^{\times}$
    \vspace{1mm}
    \Ensure $a = A(\gamma).\phi^{-1} \, \text{mod}\, p$
    \vspace{1mm}
    \State $A \leftarrow \text{RedCoeff}(A)$
	\vspace{1mm} 
	\State $a \leftarrow 0$
	\vspace{1mm}     
    \For{$i=n-1 \dots 0$}
    \vspace{1mm}
    \State $a \leftarrow (a\gamma + A_i) \, \text{mod}\, p$
    \vspace{1mm}
    \EndFor
    \vspace{1mm}
    \State return $a\beta\ \text{mod}\, p$
  \end{algorithmic}
\end{algorithm}

\section{Conclusion}
In this paper, we generalized some results in \cite{amns_12} to a bigger set of polynomials $E$. We presented a complete set of algorithms for arithmetic and conversion operations in the AMNS and shown how to generate all parameters needed for these algorithms. Our implementations have shown that AMNS allow to perform modular operations more efficiently than well known librairies like GNU MP and OpenSSL. They can even be much more efficient if their high parallelisation capability is used.
Finally we brought some arguments and elements to  point out that
 AMNS should be considered as a potential countermeasure in the context of side channel attacks.
 
%
%
\bibliographystyle{splncs04}
\bibliography{biblio}

\vspace{5mm}

\appendix

\section{Example of AMNS for the NIST recommended 521-bits prime integer for ECC}
\label{annex:nist-521}
\begin{itemize}
	\item $p = 2^{521} - 1$
	\item $n = 10$
	\item $\lambda = 2$
	\item $\rho = 2^{58}$
	\item $\gamma = 2^{469}$
	\item $E(X) = X^{10} - 2$
	\item $M(X) = 2^{52}.X - 1$
	\item $M'(Y) = 2^{52}.Y + 1$
\end{itemize}

\section{List of prime numbers used for table \ref{tab:nb-amns}}
\label{annex:prime-list}
\begin{itemize}
\item $p192 = 0xE06F20509A52674228D4F0701A08EB3B08C1\\714F0A93F719$ \\
\item $p224 = 0xE886C555B533B33B037F4F356CB97E00B56\\0DD1B5A9C252CCEAF301B$ \\
\item $p256 = 0x8FFB5E3E4BD153C220C28FDBA587F9C23\\D454DBE31C17D0B44462E26684B46E5$ \\
\item $p384 = 0xF3D1CD992E8EA43D29612F131C05A03215F\\247E92951AB3D741FEA820526FD185CDBEC7AEFC3\\1F75BEA2D2F4F43D1547$ \\
\item $p521 = 0x15683E5BD61DA4E3A10A95DE122E3B015F\\AC3F355F6360F33FA19D036CA02897BAF3D615ADA\\F6508A1E5B325B0345F39505A7B84ED01A8F913CA0\\D6395A9E135BE3$
\end{itemize}

\section{Examples of AMNS for different primes}
\label{annex:amns}
In this section, we give some examples of the AMNS we used in subsection \ref{subsec:pr-perf-mem} for our numerical experimentation.
All these AMNS have the common parameter $\phi = 2^{64}$.
\subsection{AMNS 1: 192-bits prime number.}
\begin{itemize}
	\item $p = 0xE06F20509A52674228D4F0701A08EB3B08C171\\4F0A93F719$
	\item $n = 4$
	\item $\lambda = -1$
	\item $\rho = 2^{51}$
	\item $\gamma = 0x7AB09A124AA5065B2E20034E0D0FE3D0A\\
	5F2A276C33E2515$
	\item $E(X) = X^4 + 1$
	\item $M(X) = 0x4B3D12868945.X^3 - 0x924097D431D8.X^2 + 0x39B561D62725.X + 0xC580DC0A05E3$
	\item $M'(Y) = 0x6E2B6D9BAF275F4F.Y^3 + \\
	0x8F59D05762288B18.Y^2 + 0x69A1F846105E39CF.Y \\
	+ 0xBEDE53CF67CF2747$
\end{itemize}

\subsection{AMNS 2: 224-bits prime number.}
\begin{itemize}
	\item $p = 0xE886C555B533B33B037F4F356CB97E00B560\\DD1B5A9C252CCEAF301B$
	\item $n = 4$
	\item $\lambda = -2$
	\item $\rho = 2^{60}$
	\item $\gamma = 0x64892FE7A2B9E28E496952B025FE138C22382\\6010F31C90E9354AFEF$
	\item $E(X) = X^4 + 2$
	\item $M(X) = -0x6A2300C9FAC40E.X^3 \\
	- 0xE12EC6DCB579A6.X^2 - 0x272839DE2E827E.X \\
	- 0x43419ADAFCFB61$
	\item $M'(Y) = 0x7D4F705603D9CE42.Y^3 \\
	+ 0xE0922181D0445FA6.Y^2 + 0x5A4FA29325678B32.Y \\
	+ 0xDDDE890AB0458D59$
\end{itemize}

\subsection{AMNS 3: 256-bits prime number.}
\begin{itemize}
	\item $p = 0x8FFB5E3E4BD153C220C28FDBA587F9C23D454\\DBE31C17D0B44462E26684B46E5$
	\item $n = 5$
	\item $\lambda = 2$
	\item $\rho = 2^{55}$
	\item $\gamma = 0x42559355ED8CAAA92688CE0A9322458EE4372\\4D997327755F385B1901F25E507$
	\item $E(X) = X^5 - 2$
	\item $M(X) = -0x7F360937497B.X^4 - 0x45FB30302B149.X^3 - 0x1910C5989E6B8.X^2 - 0x28750BDCB9CA3.X +\\ 0x3935AF11550E5$
	\item $M'(Y) = 0x6AC1B8BE18685FC6.Y^4 + \\0x1E8123E1FA66C4B2.Y^3 + 0x5C7430F9C82014D1.Y^2 + 0x33A24848D6BF6427.Y + 0xCC7C0CE54B67A803$
\end{itemize}

\subsection{AMNS 4: 384-bits prime number.}
\begin{itemize}
	\item $p = 0xF3D1CD992E8EA43D29612F131C05A03215F24\\7E92951AB3D741FEA820526FD185CDBEC7AEFC31\\F75BEA2D2F4F43D1547$
	\item $n = 7$
	\item $\lambda = 2$
	\item $\rho = 2^{59}$
	\item $\gamma = 0xA5C4FB2BBF7D447D0E58D14E3F440AD5C7\\A0BB773BCFA856914ED875B1A8B3DD5C6327E24B\\34890BDA7782DE3050EEC4$
	\item $E(X) = X^7 - 2$
	\item $M(X) = 0x2B70420C25B6F9.X^6 + \\0x27597E8FAEFBA6.X^5 + 0x2A259AA4E719E1.X^4 + \\0x12391F5D00D4A7.X^3 - 0x26AC55039EACFD.X^2 + \\0x2747CE657C0F2D.X - 0x426A85C33ACE17$
	\item $M'(Y) = 0x36E06AB70DC02E0C.Y^6 + \\0x91EC3470F30AB1DD.Y^5 + 0x521BCB522168C88C.Y^4 + 0x51579EF6AC4A01C8.Y^3 + 0x7145B435BA15791A.Y^2 + 0xCCD28607261C6227.Y + 0x4E6A294F1FBE2093$
\end{itemize}

\subsection{AMNS 5: 521-bits prime number.}
\begin{itemize}
	\item $p = 0x15683E5BD61DA4E3A10A95DE122E3B015FAC\\3F355F6360F33FA19D036CA02897BAF3D615ADAF6\\508A1E5B325B0345F39505A7B84ED01A8F913CA0D6\\395A9E135BE3$
	\item $n = 10$
	\item $\lambda = -2$
	\item $\rho = 2^{57}$
	\item $\gamma = 0x3BEB85F1AC84420C044C472B8845A1896C68A\\CD6C78773C9392B6CE871027BD5C333EF238A11733\\384E0A7318139218D99ADDCBB39694C1207938B6CA\\6789BC3B1$
	\item $E(X) = X^{10} + 2$
	\item $M(X) = -0x3D52F259CF52C.X^9 - \\0x2F155A2F83CC6.X^8 + 0x3C5398A0AA3D2.X^7 - \\0x6161944D2155C.X^6 + 0x92266960FE012.X^5 - \\0x68DFAA2817992.X^4 - 0x996D8B98C7860.X^3 - \\0x31E83951B9F38.X^2 + 0x3E716C4C0B2A4.X +\\ 0x3304421CB90FD$
	\item $M'(Y) = 0xBA9CFB5216CEA3CC.Y^9 + \\0x4DD219C801C0DD06.Y^8 + 0x10DEC022F71CC8F2.Y^7 + 0x199161BB290DEE2C.Y^6 + 0x924D10687452E482.Y^5 + \\0x7F6A883FEED1B396.Y^4 + 0x6923B242682C1CA0.Y^3 + 0x76FA75CEF1B36AC8.Y^2 + 0xBD1EDFD16FA95474.Y + 0xC7E79022CD8CD813$
\end{itemize}

\section{Examples of AMNS for the same prime}
\label{annex:amns-ex-sm-p}
Common parameters: 
\begin{itemize}
	\item $p = 2^{255} + 95$
	\item $\phi = 2^{64}$
\end{itemize}
\subsection{AMNS 1.}
\begin{itemize}
	\item $n = 5$
	\item $\lambda = 2$
	\item $\rho = 2^{55}$
	\item $\gamma = 0x4A11EC963214E75587B184AF9B09E8871D0DF599\\1483661DE2FF6BB1E251199C$
	\item $E(X) = X^5 - 2$
	\item $M(X) = -0x28AE865829ED0.X^4 -\\ 0x3B47735E8CB55.X^3 + 0x1337D2969BC11.X^2 - \\0x46647D3BC6C24.X + 0x2B2A32D7CA88B$
	\item $M'(Y) = 0x8FAC8FCFD4AA8587.Y^4 +\\ 0xB10CC0B3B58C223.Y^3 + 0xCA57491651A44CBB.Y^2 + 0x4E154808B257394C.Y + 0x7C5906E698B85DD$
\end{itemize}
\subsection{AMNS 2.}
\begin{itemize}
	\item $n = 5$
	\item $\lambda = -3$
	\item $\rho = 2^{56}$
	\item $\gamma = 0x1EBF5A56EC92F9F46C7F0870E5E3702D3E83\\83DEAF56E4B4C3D368BD0BF3BD40$
	\item $E(X) = X^5 + 3$
	\item $M(X) = -0x258AA3DB7ADC.X^4 -\\ 0x1C961F979254D.X^3 + 0x1D9EAFCB6057C.X^2 -\\ 0x3CE080AECD314.X - 0x539D41F2093E8$
	\item $M'(Y) = 0x6F5067DF289E2148.Y^4 +\\ 0x4D82701329D99964.Y^3 + 0x1194DEB36C42D649.Y^2 + 0x823B9BE066BDC6EC.Y + 0x1B04ECB8AF0D910C$
\end{itemize}
\subsection{AMNS 3.}
\begin{itemize}
	\item $n = 6$
	\item $\lambda = 2$
	\item $\rho = 2^{47}$
	\item $\gamma = 0x27CDB601B497003ABAE910DB0E0031133262F7\\B71DA49112F58965B98FC4930D$
	\item $E(X) = X^6 - 2$
	\item $M(X) = 0x1BDD53B3E8.X^5 +\\ 0x3817CA92D94.X^4 - 0x365442524AC.X^3 +\\ 0xE287722432.X^2 - 0x2846C0EFDE3.X +\\ 0x322D02D3281$
	\item $M'(Y) = 0xA2133E675175BDD3.Y^5 +\\ 0xB61E25558E691783.Y^4 + 0x89E8235276FDFBCB.Y^3 + 0x59D0C2AAB489D4AF.Y^2 + 0xF4592CC3EF0FD023.Y + 0xD176E217D4A7BD85$
\end{itemize}

\section{Integers structures in GNU MP and OpenSSL}
\label{annex:memory}
In this section, we give the integer structures in GNU MP and OpenSSL. These are the structures we used to compute memory cunsomptions in table \ref{tab:mem}. \\\\
GNU MP mpz\_t structure : 
\begin{verbatim}
typedef struct
{
  int _mp_alloc;	  /* Number of *limbs* allocated and pointed to by the '_mp_d' field. */
  int _mp_size;		   /* abs(_mp_size) is the number of limbs the last field points to. If _mp_size is negative this 
                      is a negative number. */
  mp_limb_t *_mp_d;		/* Pointer to the limbs. */
} __mpz_struct;
\end{verbatim}
In GNU MP mpz\_t structure, there are 2 integers of type \textit{int} and an array of type \textit{mp\_limb\_t}. On the computer we used for our tests (see features at subsection \ref{subsec:pr-perf-mem}), \textit{int} is 32 bits wide and \textit{mp\_limb\_t} is 64 bits wide. In the computation of memory consumption, we considered the 2 integers of type \textit{int} as one integer of 64 bits. For more details, see : \url{https://gmplib.org/manual/Integer-Internals.html}. \\\\
OpenSSL bignum\_st structure : 
\begin{verbatim}
struct bignum_st
        {
        BN_ULONG *d; 	/* Pointer to an array of 'BN_BITS2' bit chunks. */
        int top;        /* Index of last used d +1. */
        /* The next are internal book keeping for bn_expand. */
        int dmax;       /* Size of the d array. */
        int neg;        /* one if the number is negative */
        int flags;
        };
\end{verbatim}
In OpenSSL bignum\_st structure, there are 4 integers of type \textit{int} and an array of type \textit{BN\_ULONG}, which is 64 bits wide (on our computer). In memory consumption computation, we considered the 4 integers of type \textit{int} as two 64-bit integers.  For more details, see : \url{https://www.openssl.org/docs/man1.0.2/crypto/bn\_internal.html}

\end{document}